\documentclass[preprint,12pt,authoryear]{elsarticle}

\usepackage{color}
\usepackage{amssymb}
\usepackage{amsmath}
\usepackage{amsfonts}
\usepackage{amsthm}
\usepackage{graphicx}
\usepackage{subcaption}
\usepackage{enumerate, footmisc, hyperref}

\usepackage{soul}

\numberwithin{equation}{section}
\numberwithin{figure}{section}

\newtheorem{theorem}{Theorem}
\newtheorem{lemma}{Lemma}
\newtheorem{cor}{Corollary}
\newtheorem{lem}{Lemma}

\newtheorem{defn}{Definition}

\newtheorem{prop}{Proposition}
\newtheorem*{rem}{Remark}
\newtheorem*{conj}{Conjecture}

\theoremstyle{definition}

\newcommand{\hbs}{\mathbf{\hat s}} 
\newcommand{\hs}{\hat s}

\newcommand{\cov}{\text{Cov}}

\newcommand{\R}{{\mathbb R}}
\newcommand{\Q}{\mathbb{Q}}
\newcommand{\N}{{\mathbb N}}
\newcommand{\Z}{{\mathbb Z}}
\newcommand{\bi}{\begin{itemize}}
\newcommand{\ei}{\end{itemize}}
\newcommand{\be}{\begin{enumerate}}
\newcommand{\ee}{\end{enumerate}}
\newcommand{\w}{\textbf{w}}

\reversemarginpar

\begin{document}

\title{Loss of information in feedforward social networks}
\author{ \hspace{-.3 cm}
Simon Stolarczyk$^{1}$, Manisha Bhardwaj$^{1}$, 
Kevin E. Bassler$^{1,2,3}$, 
Wei Ji Ma$^{4}$, 
Kre\v{s}imir Josi\'{c}$^{1,5,6,}$\footnote{Corresponding Author: \url{josic@math.uh.edu}}
}



\begin{abstract}
We consider model social networks in which information propagates directionally across layers of rational agents. Each agent makes a locally optimal estimate of the state of the world, and communicates this estimate to agents downstream.  When agents receive information from the same source their estimates are correlated. We show that the resulting redundancy can lead to the loss of information about the state of the world across layers of the network, even when all agents have full knowledge of the network's  structure.  A simple algebraic condition identifies networks in which information loss occurs, and we show that all such networks must contain a particular network motif.  We also study random networks asymptotically as the number of agents increases, and find a sharp transition in the probability of information loss at the point at which the number of agents in one layer exceeds the number in the previous layer. 
\end{abstract}

\maketitle

\noindent {$^{1}$Department of Mathematics, University of Houston, Houston, TX 77204}\\
{$^{2}$Department of Physics, University of Houston, Houston, TX 77204}\\
{$^{3}$Texas Center for Superconductivity, University of Houston, 77204}\\
{$^{4}$Center for Neural Science and Department of Psychology, New York University, New York, NY 10003} \\
{$^{5}$Department of Biology and Biochemistry, University of Houston, Houston, TX 77204}\\
{$^{6}$Department of BioSciences, Rice University, Houston, TX 77251}\\

\section{Introduction}\label{sec:intro}

While there are billions of people on the planet, we exchange information with 
only a small fraction of them.  How does information propagate through such social networks,
shape our opinions, and influence our decisions?  How do our interactions impact our choice of 
career or candidate in an election?  More generally,  how do we as agents in a  network aggregate noisy  
signals to infer the state of the world?  

These questions have a long history.   The general problem is not easy to describe using a tractable mathematical model, as 
it is difficult to provide a reasonable probabilistic description of the state of the world. We also lack a full understanding of how 
perception~\citep{Brunton:2013,Beck:2012}, and the information we exchange~\citep{Bahrami2010} shapes our decisions.
Progress has therefore relied on tractable idealized models  that mimic some of the main features 
of information exchange in social networks. 

Early models relied on computationally tractable
interactions, such as the majority rule assumed 
in Condorcet's Jury Theorem~\citep{condorcet}, or  local averaging assumed in the DeGroot model~\citep{degroot74}.
More recent models relied on the assumption of rational agents -- mostly Bayesian agents who use private signals and measurements (observations)
of each other's actions to maximize utility.   
These rational models are generally either sequential or iterative. In sequential models, agents are ordered and act in turn based on a private signal and the observed action of their
 predecessors~\citep{banerjee1992,bikhchandani1992}. In iterative models, agents make a measurement,
 and then iteratively exchange information with their neighbors~\citep{gale03,mossel2014}.  Sequential models have been used
 to illustrate information cascades~\citep{Bikhchandani:1998}, while  iterative models have been used to illustrate agreement and learning~\citep{Mossel:2014}.

Here we consider a sequential model in which information propagates directionally through layers of rational agents. 
The agents are part of a structured network, rather than a simple chain. 
As in the sequential model, 
we assume that information transfer is directional, and the recipient does not communicate information to its source.
This assumption could describe the propagation of information via
print or any other fixed medium.  

We assume that at each step, a layer of agents receive information 
from those in a previous layer. This is different from previous sequential models where agents  received 
information in turn from all their predecessors as in~\cite{banerjee1992, kleinberg, welch} and~\cite{bharat}. Importantly, the same information can reach an agent via multiple paths. Therefore, information received from agents in the previous layer can be redundant.  
We show that, depending on the network structure, rational agents with full knowledge of the network structure  cannot always resolve such redundancy.  As a result, an estimate of the state of the world can degrade over layers.  
We also show that network architectures that lead to information loss can amplify an agent's bias in subsequent layers. 

As an example, consider the network in Fig.~\ref{fig:basic}(a). 
We assume that the first-layer agents make measurements $x_1, x_2$, and $x_3$ of the state of the world, $s$, and that 
these measurements are normally distributed with equal variance. 
Each agent makes an estimate, $\hs^{(1)}_1, \hs^{(1)}_2$, and $ \hs^{(1)}_3,$ of $s$. The superscript and subscript refer to the layer and agent number, respectively.  An agent with global access to all first-layer estimates would be able to make the optimal (minimum-variance) estimate $\hs_\text{ideal} 
		= \frac 13 \left( \hs^{(1)}_1 +  \hs^{(1)}_2  +  \hs^{(1)}_3 \right)$ of $s$.
 
\begin{figure}[t!]

\centering
\includegraphics[scale=1]{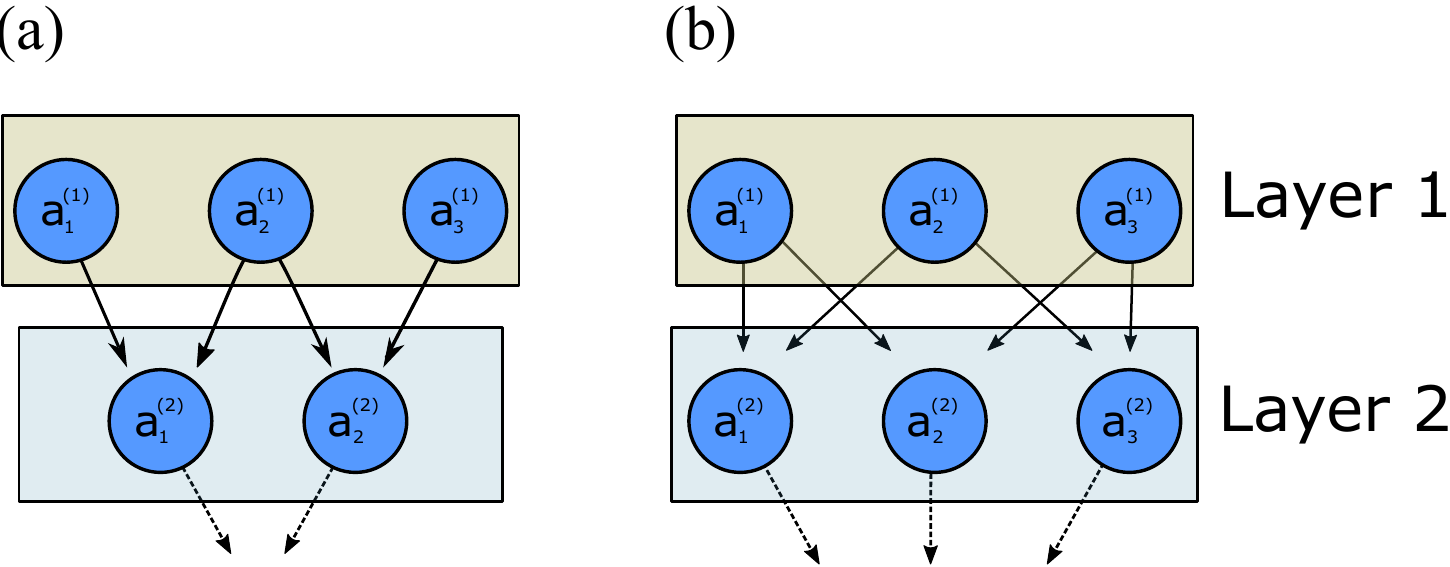}

\caption{Illustration of the general setup.  Agents in the first layer make measurements, 
$x_1, x_2$, and $x_3$, of a parameter $s$. In each layer agents make an estimate of this parameter, and communicate it to agents in the subsequent layer. 
We show that information about $s$ degrades across layers in the network in panel (a), but not in the network  in (b). }

\label{fig:basic}
\end{figure}

All agents in the first layer then communicate their estimates to one or both of the second-layer agents. 
These in turn use the received information to make their own estimates, 
$\hs^{(2)}_1 = \frac 12 ( \hs^{(1)}_1 + \hs^{(1)}_2)$ and
$\hs^{(2)}_2 = \frac 12 ( \hs^{(1)}_2 + \hs^{(1)}_3 )$.  
An agent receiving the two estimates from the second layer then takes their linear combination 
to estimate $s$.  However, in this network no linear combination of the locally optimal estimates, $\hat{s}^{(2)}_1$ and $\hat{s}^{(2)}_2,$ equals the best estimate, $\hs_\text{ideal},$ obtainable from all measurements in the first layer. Indeed,    
$$ \hs 	= \beta_1 \hs^{(2)}_1 + \beta_2 \hs^{(2)}_2 
		=  \beta_1  \left( \hs^{(1)}_1 +  \hs^{(1)}_2 \right) + \beta_2  \left( \hs^{(1)}_2 +  \hs^{(1)}_3 \right)
		\neq \hs_\text{ideal} 
		= \frac 13 \left( \hs^{(1)}_1 +  \hs^{(1)}_2  +  \hs^{(1)}_3 \right),  $$ 
with the inequality holding for any choice of $\beta_1, \beta_2$.  Moreover, assume the estimates of first-layer agents are biased, and  $ \hs^{(1)}_i = x_i + b_i$.  If the the other agents are unaware of this bias, then, as we will show, the final estimate is
 $\hs = (\frac 14, \frac 12, \frac 14) \cdot (\hs_1^{(1)}   + b_1, \hs_2^{(1)} + b_2,   \hs_3^{(1)} + b_3) = (\frac 14, \frac 12, \frac 14) \cdot \hs^{(1)} + (\frac 14, \frac 12, \frac 14) \cdot ( b_1,  b_2,  b_3).$
Thus the bias of the second agent in the first layer, $a_2^{(1)}$, has disproportionate weight in the final estimate.

In this example the information about the state of the world, $s,$ available from second-layer agents is less than that available from first-layer agents. In the preceding example the measurement $x_2$ is used by both agents in the second layer. The estimates of the two second-layer agents are therefore correlated, and  the final agent cannot disentangle them to recover the ideal estimate. We will show  that  the type of subgraph shown in Fig.~\ref{fig:basic}(a), which we call a \emph{W-motif},  provides the main obstruction to obtaining the best estimate in subsequent layers.

%
%

  \section{The Model}\label{sec:model}

We consider feedforward networks having $n$ layers and identify each node of a network with an agent.  The structure of the network is thus given by a directed graph with agents occupying the vertices.  Agents in each layer only communicate with those in the 
next layer.  For convenience, we will assume that layer $n$ consists of a single agent that receives information from 
all agents in layer $n-1$  (Fig.~\ref{fig:basic}).   This final agent in the last layer therefore makes the best estimate based on all the estimates in the next-to-last layer.  We will use this last agent's estimate to quantify information loss in the network.


We assume that all agents are Bayesian, and know the structure of the network. 
Every agent estimates an unknown parameter, $s \in {\mathbb R}$, but only the agents in the first layer make a measurement of this parameter. Each agent makes the best possible estimate given the information it receives and communicates this estimate to a subset of agents in the next layer. 
We also assume that measurements, $x_i,$ made by agents in the first layer are normally distributed and independent, $x_i \sim \mathcal N(x | s, \sigma_i^2).$
Furthermore, every agent in the network knows the variance of each measurement in the first layer,  $\sigma_i^2$. Also, for simplicity, we will assume that all agents share an improper, flat prior over $s$.  This assumption does not affect the main results.
 
An agent with access to all of the measurements, $\{x_i\}_i,$ has access to all the information available about $s$ in the network.  This agent can make an \textbf{ideal} estimate, $\hat{s}_\text{ideal} = \text{argmax}_s \; p(s | x_1, ... , x_n)$. We assume that the actual agents in the network are making locally optimal, maximum-likelihood estimates of $s$, and ask when the estimate of the final agent equals the ideal estimate, $\hat{s}_\text{ideal}$.


\paragraph{Individual Estimate Calculations}
Each agent in the first layer only have access to its own measurement, and makes an estimate equal to this measurement.  We therefore write  $\hat s_i^{(1)} = x_i$. 
We denote the  $j^{\text{th}}$ agent in layer $k$ by $a^{(k)}_j$. Each of these agents makes an estimate, $\hat s_j^{(k)}$ of $s$, using the estimates communicated by its neighbors in the previous layer.  Under our assumptions, the posterior computed by any agent is normal and the vector of estimates in a layer follows a multivariate Gaussian distribution.  
As agents in the second layer and beyond can share upstream neighbors, the covariance between their estimates is typically nonzero. We show that under the assumption that the variance of the initial measurements and the structure of the network is known to all agents, each agent knows the full joint posterior distribution over $s$ for all agents it receives information from.

\paragraph{Weight Matrices}
We define the connectivity matrix $C^{(k)}$ for $1 \leq k \leq n-1$ as,

\begin{equation} \label{def:connection_matrix}
C^{(k)}_{ij} = 
\begin{cases} 	1, & \text{if } a_j^{(k)} \text{ communicates with } a_i^{(k+1)} \\ 
				0, & \text{otherwise.} 
				\end{cases}
\end{equation} 
An agent receives a subset of estimates from the previous layer determined by this connectivity matrix. The agent then uses this information to make its own, maximum-likelihood estimate of $s$.  By our assumptions, this estimate will  be a linear combination of the communicated estimates~\citep{kay}. Denoting by $\hbs^{(k)}$ the vector of estimates in the $k^{\text{th}}$ layer, we can therefore write $\hbs^{(k + 1)}_i  = \w_i^{(k+1)} \cdot \hbs^{(k)}$, and
$$ \hbs^{(k+1)} = W^{(k+1)} \hbs^{(k)}.$$
Here $W^{(k+1)}$ is a matrix of weights applied to the estimates in the $k^{\text{th}}$ layer.

\paragraph{Weighting by Precision}

We can write $\hbs^{(1)} = W^{(1)} \mathbf{x}$ where $W^{(1)}$ is the identity matrix and $\mathbf{x}$ is the vector of measurements made in the first layer.  We assume that all measurements  have finite, nonzero variance. Defining 
$w_i := \frac1{\sigma_i^2}$, we can calculate $W^{(2)}$ entrywise: $w^{(2)}_{ij}$ is  0 if agent $a^{(2)}_i$ does \emph{not} communicate with  $a^{(1)}_j$. Otherwise 
$ w^{(2)}_{ij} = \frac{w_j^{(1)}}{\sum_{k \rightarrow i} w_k^{(1)} }
$, where the sum is taken over all agents in the first layer that communicate with agent $a^{(2)}_i$.
Therefore,
\begin{equation} \label{Eq:secondlayer_estimates}
\hbs^{(2)} = W^{(2)} \; \hbs^{(1)} = W^{(2)} W^{(1)} \mathbf{x}\;.
\end{equation}

\paragraph{Covariance Matrices}
The estimates in the second layer and beyond can be correlated. Let $L_k$ be the number of agents in the  $k^{\text{th}}$ layer and for $2 \leq k \leq n -1$ define 
$\Omega^{(k)} = (\xi^{(k)}_{ij})$ as the $L_k \times L_k$ covariance matrix of estimates in the $k^{\text{th}}$ layer,
$$\xi^{(k)}_{ij}
= \cov(\hs^{(k)}_i, \hs^{(k)}_j) .$$
When all of the weights are known, we have 
\begin{equation} \label{E:weights}
 \hbs^{(k)} 
= W^{(k)} \hbs^{(k-1)} 
= W^{(k)} W^{(k-1)} \hbs^{(k-2)} 
= \dots 
= \left( \prod_{l = 0}^{k-2} W^{(k-l)} \right) \hbs^{(1)} .
\end{equation} 
The $i^{\text{th}}$ row of $ \left( \prod_{l = 0}^{k-2} W^{(k-l)} \right)$ is the vector of weights that the agent $a_i^{(k)}$ applies to the first-layer estimates, since its entries are the coefficients in $s^{(k)}_i$.

The complete covariance matrix, $\Omega^{(k)},$ can therefore be written as
\begin{align*}\addtocounter{equation}{1}\tag{\theequation}  \label{E:omega} 
\Omega^{(k)} &= \cov(\hbs^{(k)}) = \cov (W^{(k)} \hbs^{(k-1)}) = W^{(k)} \; \cov(\hbs^{(k-1)}) \; \left( W^{(k)} \right)^{\mathrm T}
 \\ &= \left(\prod_{l = 0}^{k-2} W^{(k-l)} \right) \cov(\hbs^{(1)} ) \left(\prod_{l = 0}^{k-2} W^{(k-l)} \right)^{\mathrm T}  \\
 &= \left(\prod_{l = 0}^{k-2} W^{(k-l)} \right) \text{Diag} \left(\frac 1{w_1}, ..., \frac 1 {w_{L_1}} \right)  \left(\prod_{l = 0}^{k-2} W^{(k-l)} \right)^{\mathrm T} .
\end{align*}

Now the $i^\text{th}$ agent in layer $k \geq 3$, $a_i^{(k)}$, can use $\Omega^{(k-1)}$ to calculate $\w_i^{(k)}$. If the agent is not connected to all agents in the $(k-1)^{\text{th}}$ layer, it uses the submatrix of $\Omega^{(k-1)}$ with rows and columns corresponding to the agents in the previous layer that communicate their estimates to it. We denote this submatrix $R^{(k-1)}_i$. As in \cite{Mossel2010}, we assume that we remove edges from the graph so that all submatrices $R^{(k-1)}_i$ are invertible, but all estimates are the same as in the original network.

%
%

An agent thus receives estimates that follow a multivariate normal distribution, $ \mathcal{N}( \hbs^{(k-1)}_{j \to i}, R^{(k-1)}_i)$, see \cite{kay}.   The weights assigned by agent $a_i^{(k)}$ to the estimates of agents in the previous layer are therefore  (see also \cite{Mossel2010}), 
\begin{equation} \label{E:weight} 
\tilde{\w}^{(k)}_i = \frac{\mathbf{1}^{\text{T}} \; \left( R_i^{(k-1)} \right)^{-1} } {\mathbf{1}^{\text{T}} \; \left( R_i^{(k-1)} \right)^{-1} \; \mathbf{1} } .
\end{equation}
We define $\w^{(k)}_i$ by using the corresponding entries from $\tilde{\w}^{(k)}_i$ and setting the remainder to zero. 
In the following we describe the maximum-likelihood estimate that can be made from all the estimates in
a layer.  For simplicity, we denote this final estimate by $\hat{s}$.   The following  results are standard~\citep{kay}.

\begin{prop}\label{theorem1}
The posterior distribution over $s$ of the final agent is normal with 
\begin{equation}\label{eqn_ffn_nlayer}
\hat{s} =   
\frac {\mathbf{1}^{\text{T}} \; (\Omega^{(n-1)} )^{-1} } 
		{\mathbf{1}^{\text{T}} \; 
(\Omega^{(n-1)})^{-1} \; \mathbf{1} } \hbs^{(n-1)}
 \quad 
 \text{and} 
 \quad 
 \emph{Var} \; [\hat{s}] 
 = \frac {1} {\mathbf{1}^{\text{T}} \; (\Omega^{(n-1)})^{-1} \; \mathbf{1} } 
\end{equation}
where $\Omega^{(n-1)}$ is defined by Eq.~\eqref{E:omega} and Eq.~\eqref{E:weight}. Here  $\hat s$ is the maximum-likelihood, as well as minimum-variance, unbiased estimate of $s$.
\end{prop}
It follows from Eq.~\eqref{E:weights} that the estimate of any agent in the network is a convex linear combination of 
the estimates in the first layer.

\paragraph{Examples} Returning to the example in Fig.~\ref{fig:basic}(a) we have 

$$ C^{(1)} = \begin{pmatrix}
1 & 1 & 0 \\ 0 & 1 & 1 
\end{pmatrix}
, \;
W^{(2)} = \begin{pmatrix}
\frac 12 & \frac 12 & 0 \\
0 & \frac 12 & \frac 12 
\end{pmatrix}
,\;
\Omega^{(2)} = \begin{pmatrix}
\frac 12 & \frac 14 \\
\frac 14 & \frac 12 
\end{pmatrix}
,\;
(\Omega^{(2)})^{-1} = \frac {16}3 \begin{pmatrix}
\frac 12 & -\frac 14 \\
- \frac 14 & \frac 12
\end{pmatrix}$$

The final agent applies the weights $W^{(3)} =  \begin{pmatrix} \frac 12, \frac 12\end{pmatrix}$ to the estimates from the second layer.  We  thus have the final estimate
$\hat{s} = \begin{pmatrix} \frac 14, \frac 12, \frac 14 \end{pmatrix} \cdot \hbs^{(1)}$
with $ \text{Var} \; [\hat{s}] = \frac 38$. The variance of the ideal estimate is $\frac 13$.

On the other hand, the final agent in the example in Fig.~\ref{fig:basic}(b) makes an ideal estimate:  Here  $W^{(2)} = \begin{pmatrix}
\frac 12 & \frac 12 & 0 \\ \frac 12 & 0 & \frac 12 \\ 0 & \frac 12 & \frac 12
\end{pmatrix}$, $\Omega^{(2)} = \begin{pmatrix} \frac 12 & \frac 14 & \frac 14 \\ \frac 14 & \frac 12 & \frac 14 \\ \frac 14 & \frac 14 & \frac 12 \end{pmatrix}$,  and after inverting $\Omega^{(2)}$ we see that applying a weight of $\frac 13$ to every agent in the second layer gives the ideal estimate, $\hat{s} = \begin{pmatrix}
\frac 13, \frac 13, \frac 13
\end{pmatrix}  \cdot \hbs^{(1)} $.

\begin{rem}
If the agents have a proper prior:
\begin{equation}\label{prior}
p (s | \chi, \nu) = \mathcal{N}( s | \chi, \frac 1 \nu) = \sqrt{\frac {\nu}{ 2 \pi}} \exp \left( \frac{- \nu}{2 \pi}( s - \chi)^2 \right) ,
\end{equation}  then agents in the first layer make the estimate,
$$\hat{s}_i^{(1)} = \frac{ \frac{1}{\sigma_i^2}}{\frac{1}{\sigma_i^2} + \nu} x_i +\frac{\nu}{\frac{1}{\sigma_i^2} + \nu} \chi,$$
with a similar form in the following layers.
This does not change the subsequent results as long as all agents have the same prior. Also, if each agent
in the network makes a measurement, the general ideas remain unchanged.
\end{rem}

\section{Results} \label{sec:results}

We ask what graphical conditions need to be satisfied so that the agent in the final layer makes an ideal estimate. That is, when does knowing all estimates of the agents in the $(n-1)^{\text{st}}$ layer give an estimate that is as good as possible given the measurements of all first-layer agents. We refer to a network in which the final estimate is ideal as an \textbf{ideal} network.

\begin{prop}\label{matrix_cond}
A network with $n$ layers and $\sigma_i^2 \neq 0$ for $i = 1, \dots, L_1$, is ideal if and only if the vector of inverse variances, $(w_1, ..., w_{L_1}),$
is in the row space of the weight matrix product
$(\prod_{l = 0}^{n - 3} W^{(n - 1 -l)} )$.
\end{prop}

\begin{proof} 
In this setting the ideal estimate is 
\begin{equation} \label{E:opt}
\hs_{\text{ideal}}  
	= \frac{1}{\sum_{i} w_i}\sum_{i = 1}^{L_1} w_i \hat{s}^{(1)}_i .
\end{equation} 
The network is ideal if and only if there are coefficients $\beta_j \in \R$ such that  
\begin{equation*}
\hs_{\text{ideal}}  = \sum_{j = 1}^{L_{n-1}} \beta_j \hs_j^{(n-1)}. 
\end{equation*}
Matching coefficients with Eq.~\eqref{E:opt}, we need 
\begin{equation*}
 \frac{1}{\sum_j w_j} \sum_{i = 1}^{L_1} w_i \hat{s}^{(1)}_i 
 	= \left(\beta_1, ... , \beta_{L_{n-1}}\right) \cdot \hbs^{(n-1)},
\end{equation*}	
or equivalently,	
\begin{align*}
\frac{1}{\sum_j w_j} \left(w_1, ... , w_{L_1}\right) \cdot \hbs^{(1)} 
	&= \left(\beta_1, ... , \beta_{L_{n-1}}\right) \cdot W^{(n-1)} \hbs^{(n-2)} \\	
	&= \left(\beta_1, ... , \beta_{L_{n-1}}\right) \cdot \left( \prod_{l = 0}^{n -3} W^{(n - 1 -l)} \right) \hbs^{(1)}.
\end{align*}
Equality holds exactly when $(w_1, ..., w_{L_1})$ is in the row space of $\left(\prod_{l = 0}^{n - 3} W^{(n - 1 -l)} \right)$.
\end{proof}

In particular,  a three-layer network with  $\sigma_i^2 = \sigma$ for all $i \in \{1, \dots, L_1\}$ is ideal if and only if the vector  $\vec{1} = (1, 1, ... , 1)$ is in the row space of the connectivity matrix $C^{(1)}$ defined by Eq.~\eqref{def:connection_matrix}.  We will  use and extend this observation below.

\subsection{Graphical Conditions for Ideal Networks}\label{graphical_conditions}

We say that a network contains a \textbf{W-motif} if   two agents downstream receive common input from a first-layer agent, as well as private input from two distinct first-layer agents.  Examples are shown in Fig.~\ref{fig:basic}(a) and Fig.~\ref{fig:multiW}. A rigorous definition follows. 

We will show that \textit{all networks that are not ideal contain a W-motif}. However, the converse is not true: The network  in Fig.~\ref{fig:basic}(b) contains many W-motifs, but  is ideal.  Therefore ideal networks can contain a W-motif, as the redundancy introduced by a W-motif can sometimes be resolved. Hence, additional graphical conditions determine if the network is ideal. 

\begin{figure}[h]
\begin{center}
\includegraphics[scale=.8]{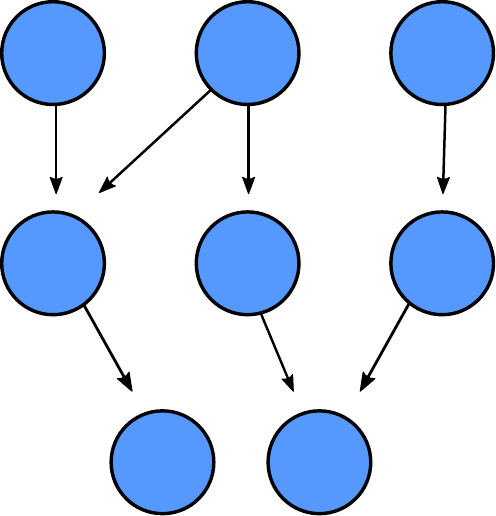}
\end{center}
\caption{A W-motif spanning three layers.}
\label{fig:multiW}
\end{figure}

As shown in Fig.~\ref{fig:multiW}, in a W-motif there is a directed path  from a single agent in the first layer to two agents in the third layer. There are also paths from distinct first-layer agents to the two third-layer agents.  This general structure is captured by the following definitions.

\begin{defn} The path matrix $P^{k l}$, $l < k$, from layer $l$ to layer $k$ is defined by, 
$$ P^{k l}_{i j} = \begin{cases} 1, & \text{if there is a directed path from agent } a_j^{(l)} \text{ to agent } a^{(k)}_i \\ 0, & \text{otherwise.}
\end{cases}$$
\end{defn}

\begin{defn}
 A network contains a W-motif if a path matrix from the first layer, $P^{k 1},$ has a $2\times3$ submatrix equal to $\begin{pmatrix}
  1 & 1 & 0 \\ 0 & 1 & 1
  \end{pmatrix}$  (modulo column permutation). Graphically, two agents in layer $k$ are connected to one common, and two distinct agents in layer $1$. 
\end{defn}

\begin{theorem}
\label{thm:overlap}
A non-ideal network in which every agent communicates its estimate to the subsequent layer must contain a W-motif. Equivalently, if there are no W-motifs, then the network is ideal.
\end{theorem}

The proof of this theorem can be found in \ref{Appendix}.  Intuitively, any agent receives estimates that are a linear combination of first-layer measurements. If there are no W-motifs, any two estimates are either obtained from disjoint sets of measurements, or the measurements  in the estimate of one agent contain the measurements  in the estimate of another.  When measurements are disjoint, there are no correlations between the estimates and thus no degradation of information.  When one set of measurements contains the other, then the estimates in the subset are redundant and can be discarded.  Therefore, this redundant information does not cause a degradation of the final estimate.

\subsection{Sufficient Conditions for Ideal Three-Layer Networks}

We next consider only  three-layer networks. This allows us to give a graphical interpretation of the algebraic condition describing ideal networks in Proposition~\ref{matrix_cond}. To do so, we will use the following corollary of the proposition.

\begin{cor} \label{Cor1}
Let $C^{(1)}$ be defined as in Eq.~\eqref{def:connection_matrix}. Then  
a three-layer network is ideal if and only if  the vector $ m \vec{1} $ is in the row space of  $C^{(1)}$ over $\Z$ for some nonzero $m \in \N$.
\end{cor}

The proof is straightforward and provided in~\ref{A2} for completeness. Note that the corollary is not restricted to the case where first-layer agents have equal variance measurements; whether the network is ideal or not depends entirely on the connection matrix $C^{(1)}$. 
The $i^{\text{th}}$ row of the matrix $C^{(1)}$ corresponds to the inputs of agent $a^{(2)}_i$,
and the sum of the $j^{\text{th}}$ column is the out-degree of agent $a^{(1)}_j$.
Therefore, Corollary~\ref{Cor1} is equivalent to the following: If each second-layer agent applies equal integer weights to all of its received estimates, then a three-layer network is ideal if and only if, for some choice of weights, the weighted out-degrees of all agents in the first layer are equal. Hence, we have the following special case:

\begin{cor}\label{cor2}
A three-layer network is ideal if all first-layer agents have equal out-degree in each connected component of the network restricted to the first two layers.
\end{cor}


%

In the connected network in Fig.~\ref{fig:basic}(a), the second agent in the first layer has greater out-degree than the others, while the agents in the first layer of the connected network in Fig.~\ref{fig:basic}(b) have equal out-degree.


Some row reduction operations can be interpreted graphically. Let  $g$ be the \textbf{input-map} which maps an agent, $a^{(2)}_i,$ to the subset of agents in the first layer that it receives estimates from.  If $g(a^{(2)}_i) \subseteq g(a^{(2)}_j)$ for some $i \neq j$, then some of the information received by  $a^{(2)}_j$ is redundant, as it is already contained in the estimate of agent $a^{(2)}_i$.  We can then reduce the network by eliminating the directed edges from $g(a^{(2)}_i)$ to $a^{(2)}_j$, so that in the reduced network $g(a^{(2)}_i) \cap g(a^{(2)}_j) = \emptyset$. This reduction  is equivalent to subtracting row $i$ from row $j$ of $C^{(1)}$ resulting in a connection matrix 
with the same row space. By Proposition~\ref{matrix_cond}, the reduced network  is ideal if and only if the original network is ideal. This motivates the following definition.

\begin{defn}
A three-layer network is said to be \textbf{reduced} if $g(a^{(2)}_i)$ is not a subset of $g(a^{(2)}_j)$ for all $1 \leq i \neq j \leq L_2$. 
\end{defn}


Reducing a network eliminates edges, and results in a simpler network structure. In a three-layer network, this will not affect the final estimate: Since reduction leaves the row space of $C^{(1)}$  unchanged, the final estimate in the reduced and unreduced network is the result of applying the same weights  to the first-layer estimates. This reduction procedure often simplifies identification of ideal networks to a counting of out-degrees (see Corollary~\ref{cor2}).

\paragraph{Example}
In Fig.~\ref{F:reduction}, we illustrate a two-step reduction of a network. In both steps, an agent (in yellow) has an input set which is overlapped by the input sets of some other agents (bolded). We use this to cancel the common inputs to the bolded agents and simplify the network. In the first step, note that the yellow agent receives input (in red) from a single first-layer agent. We use this to remove all of the other connections (in green) emanating from this first-layer agent. In the second step, we again see that the yellow agent receives input (red) that is overlapped by input to the agent next to it.  We can thus remove the redundant inputs (in green) to the bolded agent. The reduced network has 5 connected components all containing vertices with equal out-degree. Hence, this network is ideal by Corollary~\ref{cor2}.

\begin{figure}[h] 
\begin{center}
 \includegraphics[scale=.385]{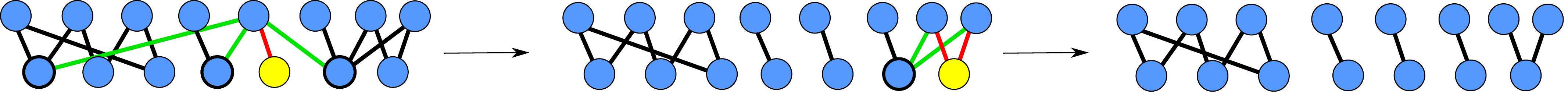}
 \end{center}
 \caption{Example of a two step network reduction.  It is difficult to tell whether the network on the left is ideal. However, after the reduction, 
 all first-layer agents in each of the five connected components have equal out-degree.  The network is therefore ideal. }
 \label{F:reduction}
\end{figure}

\subsection{Variance and Bias of the Final Estimate}
We next consider how the variance and bias of the estimate in layer $n$ depend on the network structure. By definition, the variance of the ideal estimate is $\text{Var}( \hs ) = \left( \sum_{i=1}^{L_1} w_i \right)^{-1}$. 
Therefore, as the size of the network increases, the final estimate in an ideal network is \emph{consistent}: As the number of measurements increases the final estimate converges in probability to the true value of $s$~\citep{kay}.   We next show that the final estimate in non-ideal networks is not necessarily consistent. We also show that biases of certain first-layer agents can have a disproportionate impact on the bias of the final estimate.

\begin{figure}[h]
\centering 
\includegraphics[scale=1]{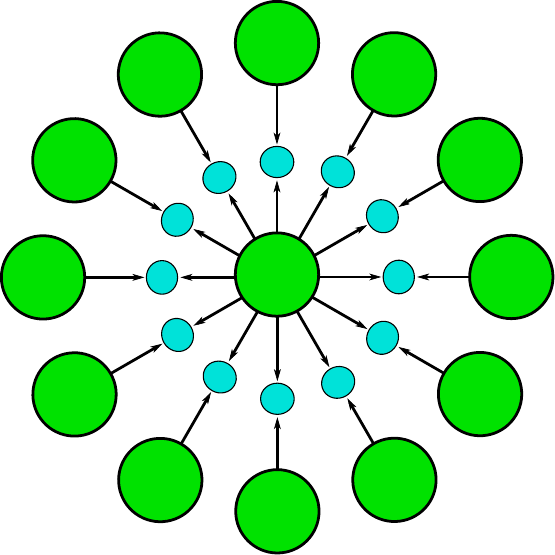}

\caption{Example of a network with an inconsistent final estimate. The green and blue nodes represent agents in the first and second layer, respectively. Each second-layer agent receives input from the common, central agent and a distinct first-layer agent.}
\label{Fig:Ring}
\end{figure}

\paragraph{Variance Maximizing Network Structure}
 Fig.~\ref{Fig:Ring} shows  an example of a  network structure for which the 
variance of the final estimate  converges to a positive number as the number of agents in the first layer increases.  
We assume that all first-layer agents make measurements with unit variance. We will show that as  the number of agents in both layers increases, the variance of the final estimate approaches $1/4$. If there are $n+1$ agents in the first layer and $n$ agents in the second layer, we have
$$C^{(1)} = \begin{pmatrix}
1 & 1 & 0 & \dots & 0 \\ 
1 & 0 & 1 & \dots & 0 \\
\vdots & \vdots & 0 & \ddots & 0 \\
1 & 0 & 0 & \dots & 1
\end{pmatrix}
,\;
W^{(2)} = \begin{pmatrix}
\frac 12 & \frac 12  & 0 & \dots & 0 \\ 
\frac 12  & 0 & \frac 12  & \dots & 0 \\
\vdots & \vdots & 0 & \ddots & 0 \\
\frac 12  & 0 & 0 & \dots & \frac 12 
\end{pmatrix}
,\;
\Omega = \begin{pmatrix}
2 & 1 & \dots & 1 \\
1 & 2 & \dots & 1 \\
\vdots & \vdots & \ddots & 1 \\
1 & 1 & \dots & 2 
\end{pmatrix} .
$$


The inverse of $\Omega$ can be explicitly computed. But note that $\Omega$ is a circulant matrix, and hence so is its inverse. This means that every row of the inverse adds to the same number, and so
$W^{(3)} = \begin{pmatrix}
\frac 1n , \dots, \frac 1n
\end{pmatrix}$.
This gives the estimate $\hat{s} = W^{(3)}W^{(2)} \hbs^{(1)} = 
\begin{pmatrix}
\frac 12, \frac 1{2n}, ...,  \frac 1{2n}
\end{pmatrix} \cdot  \hbs^{(1)}. $
Therefore, the estimate of the central agent (which communicates with all agents in the second layer) receives a much higher weight than all other estimates from the first layer. The variance of this estimate is equal to the sum of the squares of the weights,
$$
\text{Var}(\hat{s}) = \frac 14 +  \frac 1{4n}. 
$$
Hence, the final estimate is not consistent, as its variance remains positive as the number of first-layer agents diverges. Given a restriction on the number of second-layer agents, we show that this network leads to the highest possible variance of the final estimate:

\begin{prop} \label{prop:maxvar}
The final estimate in the network in Fig.~\ref{Fig:Ring} has the largest variance among all three-layer networks with a fixed number $n \geq 4$ of first-layer, and $m \geq n - 1$  second-layer agents, assuming that every first-layer agent makes at least one connection. 
\end{prop}

The idea  of the proof is to limit the possible out-degrees of the agents in the first layer and show that the structure in Fig.~\ref{Fig:Ring} has the highest variance for this restriction.
The proof is provided in~\ref{A:maxvar}.

In general, we conjecture that for the final estimate to have large variance, some agents upstream must have a disproportionately large out-degree, with the remaining agents making few connections.  On the other hand, as the in-degree of a second-layer agent increases, the variance of its estimate shrinks. 
Thus when a few agents communicate information to many, 
the resulting redundancy  is difficult to resolve downstream.  But when downstream agents receive many estimates, we expect the estimates to be good. We next show that the biases of the agents with the highest out-degrees can have an outsized influence on the estimates downstream.

%
\paragraph{Propagation of Biases}  We next ask whether the bias of the final estimate can 
remain finite in the limit of infinitely many measurements. We assume constant, additive biases, $ \hat{s}_i^{(1)} =  x_i + b_i, $ with the constant bias, $b_i,$ unknown to  agents downstream.  Since all 
estimates in the network are convex linear combinations of first-layer measurements, the final estimate will have the form
\begin{equation} \label{eqn:bias}
\hat{s} = \sum \alpha_i \left( x_i + b_i \right) = \sum \alpha_i  x_i  + \sum \alpha_i  b_i,
\end{equation} 
 and thus will have finite bias bounded by the maximum of the individual biases. 

We have provided examples of network structures where the estimate of a first-layer agent was given higher weight than others, even when all first-layer measurements had equal variance. Eq.~\eqref{eqn:bias} shows that this agent's  bias will also be disproportionately represented in the bias of the final estimate.  Indeed, in the example in Fig.~\ref{fig:basic}(a), the estimate of second agent in first layer has weight $\frac 12$, and its bias will have twice the weight of the other agents in the final estimate. Similarly, the bias of the central agent in Fig.~\ref{Fig:Ring} will account for half the bias of the final estimate as $n \to \infty$.  Thus even if
the biases, $b_i$, are distributed randomly with zero mean, the asymptotic bias of the final estimate does not always disappear as
the number of measurements increases.

More generally, networks that contain W-motifs  can result in biases of first-layer agents with disproportionate impact on the final estimate. As with the variance, we
conjecture that the bias of agents that communicate their estimates to many agents downstream will be disproportionately represented in the final estimate.  Equivalently, if the network contains  agents that receive many estimates, we expect the bias of the final estimate to be reduced.

\subsection{Inference in random feedforward networks}

We have shown that networks with specific structures can lead to inconsistent and asymptotically biased final estimates. We now consider networks with randomly and independently chosen connections between layers. Such networks are likely to contain many W-motifs, but it is unclear whether these motifs are resolved and whether the final estimate is ideal. We will use results of random matrix theory to  show that there is a sharp transition in the probability that a network is ideal when the number of agents from one layer exceeds that of the previous layer~\citep{bollobas}.  

We assume that connections between agents in different layers are random, independent and made with fixed probability, $p$. We will use the following result of~\cite{Komlos68}, also discussed by \cite{bollobas}:

\begin{theorem}[Komlos] \label{komThm}
Let $\xi_{ij}$, $i,j=1, \ldots, n$ be i.i.d. with non-degenerate distribution function $F(x)$.  Then 
the probability that the matrix $X = (\xi_{ij}) $ is singular converges to 0 with the size of the
matrix,
$$ \lim_{n \to \infty} P( \det X = 0 ) = 0.
$$
%
\end{theorem}

\begin{cor} \label{rand3}
For a three-layer network with independent, random, equally probable ($p = 1/2$) connections from first to second-layer, as the number of agents $L_1$ and $L_2$ increases,
$$\frac{L_1}{L_2} \leq 1 \implies P( \hs = \hat{s}_\text{ideal} ) \to 1,$$ and  
$$\frac{L_1}{L_2} > 1 \implies P( \hs = \hat{s}_\text{ideal}) \to 0.$$
\end{cor}

\begin{figure}
 \includegraphics[width = 6.4 in]{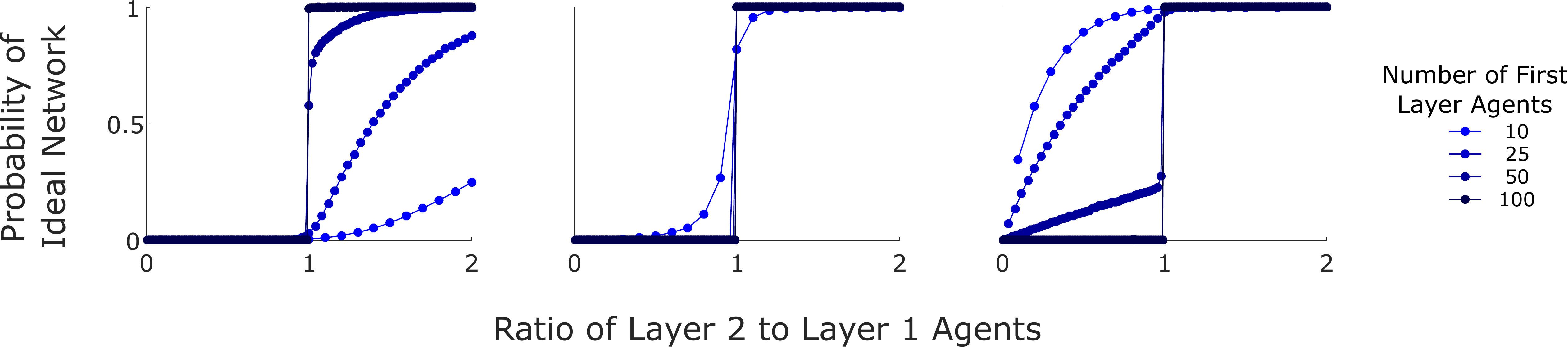}
\caption{The probability that a random, three-layer network is ideal for connection probabilities $p  =$ 0.1 (left), 0.5 (center) , and 0.9 (right). In each panel, the different curves correspond to different, but fixed numbers of agents in the first layer.  The number of agents in the  second layer is varied.  There is a sharp transition in the probability that a network is ideal when  the number of agents in the the second layer exceeds the number in the first. } 
\label{RandNetSims}
\end{figure}

The proof is given in~\ref{A:rand}. The same proof works when $L_1/L_2 \leq 1$ and the probability of a connection is arbitrary, $p \in (0,1]$. 
We conjecture that the result also holds for $L_1/L_2 > 1$ and arbitrary $p$,  but the present proof relies on the assumption that $p = 1/2$.
Fig.~\ref{RandNetSims} shows  the results of simulations which support this conjecture: The different panels correspond to different connection probabilities, and the curves to different numbers of agents in the first layer. As the number of agents in the second layer exceeds that in the first, the probability that the network is ideal approaches 1 as the number first-layer agents increases. With 100 agents in the first layer, the curve is approximately a step function for all connection probabilities we tested.

\paragraph{More than 3 Layers} We conjecture that a similar result holds for networks with more than three layers:

\begin{conj} \
 For a network with $n$ layers with  independent, random, equally probable connections between consecutive layers, as the total number of agents increases,
 $$L_k \leq L_{k+1} \text{ for } 1 \leq k < n-1 \implies P( \hs = \hat{s}_\text{ideal} ) \to 1$$ and
 $$L_1 > L_k \text{ for some } 1 < k < n \implies  P( \hs = \hat{s}_\text{ideal} ) \to 0.$$
\end{conj}

Fig.~\ref{multiRandSims} shows the results with four-layer networks with different connection probabilities across layers. The number of agents in the first and second layers are equal, and we varied the number of agents in the third layer.  The results support our conjecture.

With multiple layers, if $L_1 > L_2$ then the network will not be ideal as in the limit the estimate of $s$ will not be ideal already in the second layer by Corollary~\ref{rand3}. If the number of agents does not decrease across layers, we conjecture that the probability that information is lost across layers is small when the number of agents is large. Indeed, it seems reasonable that the products of the random weight matrices will be full rank with increasing probability allowing us to apply Proposition~\ref{matrix_cond}. However, the entries in these matrices are no longer independent, so classical results of random matrix theory no longer apply.
  
\begin{figure}[h]
\includegraphics[width = 6.4 in]{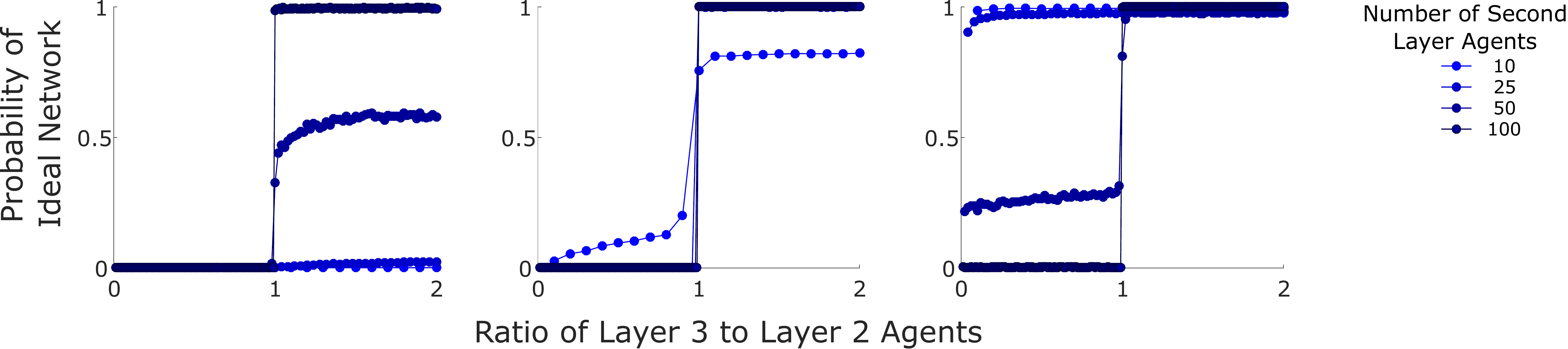}
\caption{The probability that a random, four-layer network is ideal for connection probabilities $p  =$ 0.1 (left), 0.5 (center) , and 0.9 (right). Each curve corresponds to equal, fixed numbers of agents in the first two layers, with a changing number of agents in the third layer. 
}
\label{multiRandSims}
\end{figure}


%
%


\section{Conclusion}\label{sec:conclusion}

We examined how information about the world propagates through layers of agents.  We assumed that at each step,
a group of agents makes an inference about the state  of the world from information provided by their predecessors.
The setup is related, but  different from information cascades where a chain of rational agents make decisions in
turn \citep{banerjee1992, kleinberg, welch, bharat}, or recurrent networks where agents exchange information iteratively~\citep{mossel2014}.

We translated the question about whether the estimate of the state of the world degrades across layers in the network to a  simple algebraic condition. This allowed us to use results of random matrix theory in the case of random  networks, find equivalent networks through an intuitive reduction process, and identify a class of networks in which estimates do not degrade across layers, and another class in which degradation is maximal.

Networks in which estimates degrade across layers must contain a W-motif. This motif introduces redundancies in the
information that is communicated downstream and may not be removed.  Such redundancies, also known as ``bad correlations,'' are known to limit the information that can be decoded from neural responses~\citep{MorenoBote:2014,Bhardwaj:2015}.   This  suggests that agents with 
 large out-degrees and small in-degrees can hinder the propagation of information, as they introduce redundant information 
in the network.  On the other hand, agents with
large in-degrees integrate information from many sources, which can help improve the final estimate.  However, the detailed
structure of a network is important:  For example, an agent with large in-degree in the second layer can have a large out-degree without hindering the propagation of information as it 
has already integrated most available first-layer measurements.

To make the problem tractable, we have made a number of simplifying assumptions. We made the strong assumption that agents have full knowledge of the network structure. Some agents may have to make several calculations in order to make an estimate, so we also do not assume bounded rationality \citep{goyal98}. This is unlikely to hold in 
realistic situations. Even when making simple decisions, pairs of agents are not always rational~\citep{Bahrami2010}: When two  agents each make a measurement with different variance, exchanging information can degrade the better estimate. 

The assumption that only agents in the first layer make a measurement is not crucial.  We can obtain similar results if
all agents in the network make independent measurements, and the information is propagated directionally, as we assume here.
However, in such cases, the confidence (inverse variance of the estimates) typically becomes unbounded across layers.

\section{Acknowledgments}

Funding: This research was supported by NSF-DMS-1517629 (SS and KJ),  NSF/NIGMS-R01GM104974 (KJ), NSF-DMR-1507371 (KB), and NSF-IOS-1546858 (KB). 

\begin{appendix}

\section{Proof of Theorem~\ref{theorem1}} \label{Appendix}

We start with the simpler case of a W-motif between the first two layers and then extend it to the general case. We begin with definitions that will be used in the proof. 


Let $g$ be the \textbf{input-map} which maps an agent to the subset of agents in the first layer that it receives information from (through some path). That is, $g( a_i^{(j)})$ is the set of agents in the first layer that provide input to $a_i^{(j)}$. It is intuitive -- and we  show it formally in Lemma~\ref{noWOrder} -- that a network contains a W-motif if each of the inputs to two agents, $A$ and $B$ are not contained in the other, and their intersection is not empty. That is,  $g(A) \not\subseteq g(B)$ and $g(B) \not\subseteq g(A),$ but $g(A) \cap g(B) \neq \emptyset $. If these conditions are met, we also say that the inputs of $A$ and $B$ have a \textbf{nontrivial intersection}. 
If $g(A) \subseteq g(B)$, we say that the input of $B$ \textbf{overlaps} the input of $A$: every agent which contributes to the estimate of $A$ also contributes to the estimate of $B$.

Similarly, we let $f$ be the \textbf{output-map} which maps an agent, $a_{i}^{(j)},$ to the set of all agents in the next, $j+1^{\text{st}}$, layer that receive input from $a_{i}^{(j)}$. We first prove a few lemmas essential to the proof of Theorem~\ref{thm:overlap}.

\begin{lem}
\label{noWOrder}
Assume a network does not contain a W-motif and there are two agents, $a_{i_1}^{(k)}$ and $a_{i_2}^{(k)}$, with $g(a_{i_1}^{(k)}) \cap g(a_{i_2}^{(k)})$ nonempty. Then  $g(a_{i_1}^{(k)})$ overlaps or is overlapped by $g(a_{i_2}^{(k)})$.
\end{lem}

\begin{proof} 
We prove the claim by contradiction. If one input does not overlap the other, then there are two distinct first-layer agents $a_{n_1}^{(1)}$ and $a_{n_2}^{(1)}$ such that $ a_{n_1}^{(1)} \in g(a_{i_1}^{(k)}) \setminus g(a_{i_2}^{(k)})$ and  $ a_{n_2}^{(1)} \in g(a_{i_2}^{(k)}) \setminus g(a_{i_1}^{(k)})$. This means 
$P^{k 1}_{ i_1 n_1 } =  P^{k 1}_{ i_2 n_2} = 1$ and $P^{k 1}_{i_1 n_2  } = P^{k 1}_{ i_2 n_1 } = 0$. Since the inputs of the agents have nonempty intersection, we also have $P^{k 1 }_{i_1  m } = P^{k 1}_{ i_2 m } = 1$ for some $m$.
Thus there is a $2 \times 3$ submatrix of $P^{k 1}$ which, up to rearrangement of the columns, is equal to $\begin{pmatrix} 1 & 1 & 0 \\ 1 & 0 & 1 \end{pmatrix}$ and the network contains a W-motif, contrary to assumption.
\end{proof}

Every agent's estimate is a convex linear combination of estimates in the first layer, given by Eq.\eqref{E:weights}.  We will use the corresponding weight vectors in the following proofs. We show that in networks without W-motifs, agents will only be receiving collections of estimates with weight vectors which pairwise either have disjoint support (nonzero indices) or the support is contained in the support of the other agent. Thus, with no W-motifs, no two agents have inputs with nontrivial intersection. The next two lemmas will allow us to easily calculate the estimates of such agents.

\begin{lem}	 \label{L:est}
Let $r,s,t$ be positive integers, $w_i = \sigma_i^{-2}$, and consider three weight vectors applied by three agents in layer $k$, $a_1^{(k)},a_2^{(k)},$ and $a_3^{(k)}$,  to the estimates of the first layer:
\begin{align*}
v_1 &= \left( \frac{w_1}{\sum_{i = 1}^r w_i},\dots, \frac{w_r}{\sum_{i = 1}^r w_i}, 0, \dots ,0\right)\\
v_2 &= \left( \frac{w_1}{\sum_{i = 1}^{r+s} w_i},\cdots, \frac{w_{r + s}}{\sum_{i = 1}^{r+s} w_i},  0, \dots ,0 \right)\\
v_3 &= \left( 0, \dots,0, \frac{w_{r + s +1}}{\sum_{i = r + s + 1}^{r+s+ t} w_i}, \dots, \frac{w_{r + s +t}}{\sum_{i = r + s + t}^{r+s+ t} w_i},  0, \dots ,0 \right).
\end{align*}
An agent $a_i^{(k+1)}$  in $f( a_1^{(k)}) \cap f( a_2^{(k)}),$ but not in $f( a_3^{(k)})$, will use weight vector $v_2$. An agent $a_i^{(k+1)}$  in $f( a_2^{(k)}) \cap f( a_3^{(k)}),$ but not $f( a_1^{(k)})$, will use weight vector 
$$ v_4 = \left( \frac{w_1}{\sum_{i = 1}^{{r+s+ t}} w_i}, \dots,
\frac{w_{r+s+ t}}{\sum_{i = 1}^{{r+s+ t}} w_i}, 0, ...,0 \right).$$
\end{lem}

\begin{proof}
First, consider an agent receiving the first two estimates with weights $v_1$ and $v_2$. Suppose that a fictitious agent receives a collection of estimates with weight vectors $\{z_1, ..., z_{r+s}\}$, where $z_i = (0, \dots, 0, 1, 0 , \dots, 0)$, \emph{i.e.}, each estimate equals the measurement of agent $a^{(1)}_i$. This fictitious agent can obtain any linear combination of the first $r+s$ measurements.  The linear combination with lowest variance has weights given by $v_2$.  Therefore, an agent receiving measurements corresponding to the weight vectors $v_1$ and $v_2$ cannot do better than the estimate of agent $a_2^{(k)}$ with weights given by $v_2$.

A similar argument works when estimates are received from agents $a_2^{(k)}$ and $a_3^{(k)}$. Since these two agents make locally optimal estimates based on non-overlapping sets of measurements in the first layer, the best estimate is obtained by combining the two sets of measurements.  This is precisely the estimate corresponding to the weights given by vector $v_4$.
%
%
\end{proof}

\begin{lem} \label{L:ideal}
 Suppose an agent, $a_i^{(k)},$ receives a collection of estimates such that for any pair, there is a relabeling of agents in the first layer that makes the pair look like $v_1$ and $v_2$ or like $v_2$ and $v_3$ in Lemma~\ref{L:est}. Then, up to some relabeling of the agents in the first layer, that agent will make an estimate with corresponding weight vector
 \[
  v = \left( \frac{w_1}{\sum_{i = 1}^r  w_i} , \dots,  \frac{w_r}{\sum_{i =1}^r w_i},0 , \dots,  0 \right) .
 \]

\end{lem}

\begin{proof}
Let the vectors $z_i$ be defined as in the proof of Lemma~\ref{L:est}.  Relabel the first-layer  agents so that only the first $r$ entries of the weight vector applied by agent $a_i^{(k)}$  are non-zero. Then a fictitious agent  receiving estimates with weight vectors $z_i,  1 \leq i \leq r$ can construct any estimate that agent $a_i^{(k)}$ can obtain. The optimal 
estimate of this fictitious agent has weight vector $v$.  Hence if some linear combination of the weight vectors  of estimates communicated to agent $a_i^{(k)}$ equals  $v$, this linear combination defines the best estimate.

 Then for each $j = 1, ..., r$, we can find a weight vector, $v_j$, which is nonzero in the $j^\text{th}$ entry with support that contains the support of every other weight vector which is nonzero in the $j^\text{th}$ entry. Such a vector exists by the assumption that any two vectors have disjoint support or the support of one contains the other. Therefore, we can find the weight vector with maximal support for each entry. If we take the distinct elements of $\{ v_j : 1 \leq j \leq r \}$, then these maximal weight vectors will have disjoint support that partitions the first $r$ indices.  Therefore, 

\[
v = \frac{1}{\sum_{i = 1}^r w_i} \sum_{ v_j \text{ distinct}}  \left( \sum_{i = 1, v_j^i \text{ nonzero}}^r w_i \right)  v_j ,
\]
\noindent which shows the lemma.
\end{proof}

We now state and prove the three-layer case of Theorem~\ref{thm:overlap} and then use it to finish the proof of Theorem~\ref{thm:overlap}.

\begin{prop} \label{w-motif}
If a three-layer network is not ideal and every first-layer agent communicates with at least one second-layer agent, then the network must contain a W-motif.
\end{prop} 

\begin{proof} 
Assume the network does not contain a W-motif.
Given a first-layer agent $a^{(1)}_i$, Lemma~\ref{noWOrder} says that for any two agents in $f(a^{(1)}_i)$, one agent's input must overlap the other. Two second-layer agents thus receive estimates with input sets where one overlaps the other, or the sets do not intersect.  Thus the set of weight vectors in the second layer satisfies the assumptions of Lemma~\ref{L:ideal}. As all agents from the first layer communicate with the final agent, the network is ideal.
\end{proof}

To obtain the proof of Theorem~\ref{thm:overlap}, we  use induction with Proposition~\ref{w-motif} as a base case.

\begin{proof}[Proof of Theorem \ref{thm:overlap}.]
Assume the network has $n$ layers, there are no W-motifs, and every agent (except those in the first layer) receives input from at least one other agent. 
Lemma~\ref{noWOrder} implies that in the second layer each pair of agents has either disjoint input or one overlaps the other. Thus in the third layer, by relabeling the agents, each agent makes an estimate with weight vector of the form: $\frac{1}{\sum_{i = 1}^r w_i}(w_1, \dots , w_r, 0 , \dots, 0)$.

Now assume that any estimate in layer $k$ can be put in this form by relabeling the agents. Since there are no W-motifs, Lemma~\ref{noWOrder} implies that set of measurements used by agents $a^{(k)}_{i_1}$ and $a^{(k)}_{i_2}$ is disjoint or overlapping. This again allows us to apply Lemma~\ref{L:ideal}  and any agent in layer $k+1$ makes an estimate whose weight vector again has the form $\frac{1}{\sum_{i = 1}^r w_i}(w_1, \dots , w_r, 0 , \dots, 0)$. Applying the same argument to the  final agent, where every entry will be nonzero in some penultimate-layer agent's weight vector, we have that the network is ideal.

\end{proof}

\section{Proof of Corollary~\ref{Cor1}}  \label{A2}

We will show that a three-layer network is ideal if and only if $m\vec{1}$ is in the row space of $C^{(1)}$ over $\Z$ for some $m \in \N$. We do this by first showing that the network is ideal if and only if $\vec{1}$ is in the row space of $C^{(1)}$ over $\R$, and then we show that this is equivalent to  $m\vec{1}$ being in the row space of $C^{(1)}$ over $\Z$. 

By Proposition~\ref{matrix_cond}, a three-layer network is ideal if and only if $(w_1, \dots, w_{L_1})$ is in the row space of $W^{(2)}$. We claim that this is equivalent to $\vec{1}$ being in the row space of $C^{(1)}$:  Multiplying each row of $W^{(2)}$ by the common denominator of the nonzero entries gives
\[
\mathcal{R}( W^{(2)} ) = \mathcal{R} ( C^{(1)} \text{Diag}(w_1, \dots, w_{L_1}) ),
\]
where $\mathcal{R}$ denotes the row space. By definition, $\vec{1}$ is a linear combination of the rows of $C^{(1)}$ if and only if
\[
1 = \sum_{i} \beta_i C^{(1)}_{i j} , \; \; \; \forall j.
\]
This holds if and only if
\begin{equation*}
w_j = \sum_{i} \beta_i w_j C^{(1)}_{i j} , \; \; \; \forall j. \\
\end{equation*} 
The last equality is equivalent to
\[
(w_1, \dots, w_{L_1}) = \sum_i \beta_i (C^{(1)} \text{Diag}(w_1, \dots, w_{L_1}))_{i} \; \; ,
\]
which means $(w_1, \dots, w_{L_1})$ is in the row space of $W^{(2)}$. Hence, for three-layer networks, the network is ideal if and only if the vector $\vec{1}$ is in the row space of $C^{(1)}$ over $\R$.

Thus it remains to show that $\vec{1} \in \mathcal{R} ( C^{(1)})$ over $\R$ is equivalent to $\vec{1} \in \mathcal{R} ( C^{(1)})$ over $\Z$.  If  $  m \vec{1} \in \mathcal{R} ( C^{(1)})$ over $\Z$, then it is a linear combination of the rows of $C^{(1)}$ with integer coefficients. Multiplying the coefficients of this linear combination by $\frac 1 m$ shows that $\vec{1}$ is in the row space of $C^{(1)}$ and hence the network is ideal. 

If $\vec{1}$ is in the row space of $C^{(1)}$ over $\R$, then by closure of $\Q^n$ this means there is some linear combination of the rows of $C^{(1)}$ over $\Q$ which is equal to $\vec{1}$:
\[
\sum_{i = 1}^{L_2} \alpha_i C^{(1)}_i = \vec{1} , \qquad \alpha_i \in \Q.
\]

\noindent Multiplying both sides by the absolute value of the product of the denominators of the nonzero $\alpha_i$ shows that
\[
\sum_{i = 1}^{L_2} \beta_iC^{(1)}_i = m \vec{1} , \qquad \beta_i \in \Z
\]

\noindent for some $m \in \N$ and thus $m\vec{1}$ is in the row space of $C^{(1)}$ over $\Z$.

\section{Proof of Proposition~\ref{prop:maxvar}} \label{A:maxvar}

We will show that the network architecture that maximizes the variance of the final estimate for a given number of first and second-layer agents is the one shown in Fig.~\ref{Fig:Ring}. To simplify notation we write $L_1 = n$ and $L_2 = m$. 

\begin{lem} If $\mathbf{d} = (d_1, ... , d_{n})$ is the vector of out-degrees in the first layer, so $d_i = | f(a_i^{(1)}) |$, then to maximize the variance of the final estimate, $\mathbf{d}$ must equal $(m, 1, \dots, 1)$, up to relabeling. 
\end{lem}

\begin{proof}[Proof of Claim]
Given a network structure consider the na\"ive estimate:
\begin{equation} \label{E:naive}
\frac 1Z \sum_i |g(a_i^{(2)})| \hs_i^{(2)} = \frac{1}{\sum_{i j} C_{i j}^{(1)}} \sum_i C_i^{(1)} \cdot \hbs^{(1)},
\end{equation}
where $Z$ is a normalizing factor that makes the entries of the corresponding vector of weights sum to 1. This estimate can always be made and is the same as using a linear combination of estimates of agents $a_j^{(1)}$ with weights $\frac{d_i}{\sum_{j = 1}^{n} d_j}$. Thus the variance of the optimal estimate of the agent in the final layer is bounded above by the variance of the na\"ive estimate in Eq.~\eqref{E:naive}. By assumption $1 \leq d_j \leq m$ for all $j$. For the network in Fig.~\ref{Fig:Ring}, this na\"ive estimate equals the final estimate.  Thus it is sufficient to show that the na\"ive estimate has maximal variance when $\mathbf{d} = (m, 1, \dots, 1)$, up to relabeling.

The variance, $V$, of the na\"ive estimate is:
\[
V(d_1, \dots, d_n) = \sum_j \left( \frac{d_j}{\sum_{k = 1}^{n} d_k} \right) ^2 .
\]

If we treat the degrees as continuous variables then $V$ is continuous on $\mathbf{d} \in [1,m]^n$ and we can calculate the gradient of $V$ to find the critical points. 
\[
\frac{\partial V}{\partial d_i} = 
	2 \left( \frac{d_i}{\sum_k d_k} \right)  \frac{ \sum_k d_k - d_i}{\left( \sum_k d_k \right)^2} 
	+ \sum_{j \neq i} 2 \left( \frac{d_j}{\sum_{k} d_k} \right) \frac{-d_j}{\left( \sum_{k} d_k \right)^2}
\]

Setting $ \frac{\partial V}{\partial d_i} = 0$ and multiplying both sides by $\frac 12 \left( \sum_{k = 1}^{n} d_k \right)^3$ gives
\begin{align*}
0 &= d_i ( \sum_{k \neq i} d_k)
	- \sum_{j \neq i}  d_j^2 = \sum_{j \neq i} d_j (d_i - d_j).
\end{align*}
This shows that $d = k \vec{1}$ for $k = 1, \dots , m$ are the only critical points, since if there exist $\ d_i \leq d_j,$ for all $j \neq i$ and $d_i < d_k$ for some $k \neq i$ then the right hand side would be negative. These critical points are the first-layer out-degrees of ideal networks by Corollary~\ref{cor2}, hence they are minima. This implies that $V$ takes on its maximum values on the boundary.

The boundary of $[1,m]^n$ consists of points where at least one coordinate is $1$ or $m$. Since $V$ is invariant under permutation of the variables, we set $d_1$ equal to one of these values and investigate the behavior of $V$ on this restricted set.

First set $d_1 = m$. Setting $\frac{\partial V}{\partial d_i}$ to 0 on this boundary gives:
\begin{align*}
0 &= m(d_i - m) + \sum_{j \neq i, 1} d_j (d_i - d_j) 
\end{align*}
One critical point is thus $m \vec{1}$. If $d_i \leq d_j$ for $j \neq i$ and $d_i < m$ then again the right hand side would be negative. Hence $d_i = m$ for all $i$, and there are no critical points on the interior of $\{m\} \times [1,d]^{n-1}$.

Next if  $d_1 = 1$, setting $\frac{\partial V}{\partial d_i}$ to 0 on this boundary and multiplying by $-1$ gives:
\begin{align*}
0 &= 1 - d_i + \sum_{j \neq i, 1} d_j (d_j - d_i)
\end{align*}

Here a critical point is $\vec{1}$. If $d_i \leq d_j$ for $j \neq i$ and $1 < d_i < m $ then again the right hand side would be negative. Hence $d_i = 1$ for all $i$,  and there are no critical points on the interior of $\{1\} \times [1,d]^{n-1}$.
If we iterate this procedure, we see that the maximum value of $V$ must occur on the corners of the hypercube $[1,d]^n$.

Choose one of these corners, $\mathbf{c}$, and, without loss of generality, assume that the first $l$ coordinates are $m$ and the last $n - l$ coordinates are 1, $1 \leq l < n$. Then
\begin{align*}
V(\mathbf{c}) &= \sum_{j = 1}^l \left( \frac{m}{\sum_{k = 1}^{n} d_k} \right) ^2 
+  \sum_{j = l+1}^n \left( \frac{1}{\sum_{k = 1}^{n} d_k} \right) ^2 \\
&= \left(\frac{1}{lm + (n-l)}\right)^2 \left( l m^2 + (n - l) \right) \\
&= \frac{ lm^2 + n - l}{ l^2 m^2 + 2 l m (n -l) + (n-l)^2} \\
&= \frac{ l (m^2 - 1) + n}{l^2 ( m-1)^2 + l2n(m-1) + n^2}
\end{align*}
Under the assumption that $m \geq n -1$, a lengthy algebra calculation that we omit shows that this is maximized for $l = 1$. Hence the maximum value of $V$ is achieved at $(m,1,\dots,1)$, or any of its coordinate permutations.


\end{proof}

Finally, to have $\mathbf{d} = (m, 1, \dots, 1)$, one first-layer agent, $a_1^{(1)}$, communicates with all second-layer agents and every other agent has exactly one output. Since there are at least $n -1$ agents in the second layer, this means that each first-layer agent must communicate with  a distinct second-layer agent and each second-layer agent must receive input from $a_1^{(1)}$. Otherwise, some agent in the second layer would receive only the input from $a_i^{(1)}$ and thus the final estimate could use that estimate to decorrelate all of the second-layer estimates. 

So, the na\"ive estimate for an alternative network has smaller variance than the ideal estimate for the ring network in Fig.~\ref{Fig:Ring}. Hence the final estimate in any alternative network will have smaller variance. Since the only network with $\mathbf{d} = (m, 1, \dots, 1)$ is the network in Fig.~\ref{Fig:Ring}, we have shown that this structure maximizes the variance of the final estimate among all networks with $L_2 \geq L_1 - 1$.

\section{Proof of Corollary~\ref{rand3}} \label{A:rand}

Whether or not $\hat{s}_\text{ideal} = \hs$ is  determined by $C^{(1)}$. For simplicity, we drop the superscript and refer to this connectivity matrix as $C$. By our assumption, this is a random matrix with $P(C_{ij} = 0) = P(C_{ij} = 1) = 1/2$. 

First assume that there are at least as many second-layer agents as there are first-layer agents: $L_2 \geq L_1$ or $\frac{L_1}{L_2} \leq 1$. Then $C$ is a random $L_2 \times L_1$ matrix with i.i.d. non-degenerate entries that has more rows than columns. By Theorem \ref{komThm}, this means that  the $L_1 \times L_1$  submatrix formed by the first $L_1$ rows and columns is nonsingular with probability approaching 1 as $L_1, L_2 \to \infty$.
Thus the probability that the row space of  $C$ contains the vector $\vec 1$ converges to 1 with the size of the network. 

Next assume that there are fewer second-layer agents than first-layer agents, that is $L_2 < L_1$ or $\frac{L_1}{L_2} > 1$. We will show that the probability that the row space of $C$ contains $\vec 1$ goes to zero as $L_1, L_2 \to \infty$.
Since increasing the number of rows will not decrease the probability that $C$ contains a vector in its row space we assume that $L_2 = L_1 - 1$ and let $L_1 = n$:

$$ \lim_{L_1,L_2 \to \infty} P(\hs = \hat{s}_\text{ideal} ) \leq \lim_{n \to \infty} P(\vec 1 \in R(C(n-1,n)))$$

\noindent where $C(n-1,n)$ refers to the random matrix as before, and identifies that it has $n-1$ rows and $n$ columns. We first use:
\[
P(\vec 1 \in R(C(n-1,n))) \leq P( \left( \begin{matrix} \vec 1 \\ C \end{matrix} \right) \text{ is singular} )
\]
\noindent since if $\vec 1$ is the row space of $C$, then attaching that row of ones to it would create a singular matrix.

\begin{lemma} $P \left( \det (  \left( \begin{matrix} \vec 1 \\ C \end{matrix} \right) ) = 0 \right) \to 0$ as $n \to \infty$. 
\end{lemma}

We can rewrite $C = \left( \begin{matrix} B & \mathbf{v} \end{matrix} \right)$, where $\mathbf{v}$ is the $n^\text{th}$ column of $C$ and $B$ is the remaining submatrix. 
We claim 
\begin{equation}
\label{detClaim}
(  \left( \begin{matrix} \vec 1 \\ C \end{matrix} \right) )  
= -1^k \det (  \left( \begin{matrix} \vec 1 & 1\\ \tilde{B} & \vec 0 \end{matrix} \right) ) 
= -1^{k + n + 1} * \det(\tilde{B})
\end{equation} 
where $\tilde{B}$ is a random $(n-1) \times (n-1)$ matrix distributed like $C$. 
Assuming this claim, then by \cite{Komlos68} : 
$$P(\det (  \left( \begin{matrix} \vec 1 \\ C \end{matrix} = 0 \right) ) ) = P\left( \det(\tilde{B}) = 0 \right) \to 0  \quad \text{as} \quad n \to \infty.$$ 
Thus $P(\vec 1 \in R(M(n-1,n))) \to 0$ as $n \to \infty$. 

To prove the first equality in Eq.~\eqref{detClaim}, we use row operations on $\left( \begin{matrix} \vec 1 & 1\\ B & \mathbf{v} \end{matrix} \right)$:
If $v_i = 1$ then subtract the first row from the $i^\text{th}$ row, $(B_i \; v_i)$, to get a vector whose entries are all $0$ and $-1$. Then $(B_i \; v_i) \to - (\tilde{B}_i \; 0)$ where $(\tilde{B}_i \; 0)$ is a vector of entries which are again either 0 or 1 with equal probability. We do this for every row which has a 1 in its last entry and multiply the determinant a factor $-1$ and denote the number of these reductions as $k$. Since $P( C_{i j} = 0) = \frac 12$ we also have $P(\tilde{B}_{i j} = 0) = \frac 12$.

\end{appendix}

\section*{\refname}
\bibliographystyle{elsarticle-harv}
\bibliography{FFNetsBib}

\end{document}